\documentclass[11pt]{article}

\usepackage{amsmath, amsthm, amssymb}
\usepackage{fullpage}
\usepackage{algorithm}
\usepackage{algorithmicx, algpseudocode}
\usepackage{hyperref}

\newcommand{\norm}[1]{\lVert #1 \rVert}
\newcommand{\fnorm}[1]{\norm{#1}_F^2}
\newcommand{\eps}{\epsilon}
\newcommand{\del}{\delta}
\newcommand{\R}{\mathbb{R}}
\newcommand{\E}{\mathbb{E}}
\newcommand{\spn}{\text{span}}

\newcommand{\subb}[1]{_{(#1)}}
\newcommand{\err}{err}
\newcommand{\sse}{\text{SSE}}
\newcommand{\e}{\mathbf{e}}
\newcommand{\iprod}[1]{\langle #1 \rangle}

\newcommand{\poly}{\text{poly}}
\newcommand{\nvspace}{\vspace{-.1in}}
\newcommand{\cI}{\mathcal{I}}
\newcommand{\cO}{\mathcal{O}}

\newtheorem{theorem}{Theorem}
\newtheorem{lemma}[theorem]{Lemma}
\newtheorem{claim}{Claim}

\newtheorem{corr}[theorem]{Corollary}
\newtheorem{conjecture}{Conjecture}

\title{Low Rank Approximation in the Presence of Outliers}
\author{Aditya Bhaskara\thanks{School of Computing, University of Utah. {\em Email: }\textsf{bhaskara@cs.utah.edu}} \and Srivatsan Kumar\thanks{School of Computing, University of Utah. {\em Email: }\textsf{seezha@cs.utah.edu}}}
\date{}

\begin{document}

\maketitle
\begin{abstract}
We consider the problem of principal component analysis (PCA) in the presence of outliers.  Given a matrix $A$ ($d \times n$) and parameters $k, m$, the goal is to remove a set of at most $m$ columns of $A$ (outliers), so as to minimize the rank-$k$ approximation error of the remaining matrix (inliers).  While much of the work on this problem has focused on recovery of the rank-$k$ subspace under assumptions on the inliers and outliers, we focus on the approximation problem.  Our main result shows that sampling-based methods developed in the outlier-free case give non-trivial guarantees even in the presence of outliers.  Using this insight, we develop a simple algorithm that has bi-criteria guarantees.  Further, unlike similar formulations for clustering, we show that bi-criteria guarantees are unavoidable for the problem, under appropriate complexity assumptions. 
\end{abstract}
\section{Introduction}
Low rank approximation is one of the most fundamental and well-studied questions in matrix analysis.  It is widely used for dimension reduction, sketching, denoising, and more broadly to find ``structure'' in large matrices.  Usually referred to as principal component analysis (PCA) or singular value decomposition (SVD), low rank approximations are a staple tool in the analysis of large matrix data. 

We consider the problem of low rank approximation in the presence of {\em outlier} columns.  Studying the effect of adversarial outliers on known algorithms as well as the problem complexity has become a significant theme in recent research.  It is motivated by the practical importance of dealing with noise in data (both intrinsic as well as adversarial).  While principal components are often robust to small corruptions of the matrix (which is why they are useful in denoising, e.g.~\cite{McSherry}), this typically requires the noise to have a small spectral norm (see~\cite{StewartSun}).  This is an unrealistic assumption in many settings, especially when entire columns can be (arbitrarily and possibly adversarially) noisy.  

Dealing with noise has also been well studied in the statistics community (the books of~\cite{Huber, HuberNew} present many of the classic results in the area).  Robust estimators for computing parameters of distributions such as the mean, variance, etc. have been extensively studied.  More recently, questions of this nature have received a lot of attention in theoretical computer science, motivated by connections to learning distributions.  The works of~\cite{MoitraNew, VempalaNew}, as well as subsequent works (see~\cite{SteinhardtCV18} and references therein) have obtained novel guarantees on recovery under noise, using semidefinite programming and other techniques. 

Meanwhile, clustering problems have long been studied in the presence of outliers, from the point of view of approximation algorithms.  Starting with the work of~\cite{Charikar01}, and subsequent improvements (see \cite{Cluster2, Cluster1, RaviShi}), we now have a  good understanding of the approximability of clustering with outliers.  Many of these works use linear programming relaxations that help identify the outliers.  One of our motivations is to obtain a similar understanding for PCA.

\nvspace
\paragraph{PCA with outliers.}  We now define the formulation we study in this paper.  It is motivated by the work on clustering algorithms discussed above.  Informally,  given a matrix $A$ ($d \times n$), we wish to find the best rank $k$ approximation to $A$,  after throwing away $m$ columns of our choice.  This can be formally stated as the following optimization problem:
\[ \text{minimize }~ \fnorm{A- L -N}, \quad \text{subject to rk$(L) \le k$, and $N$ having at most $m$ non-zero columns.} \]
This formulation is has also been studied in the signal processing literature~\cite{Chen, CandesRobust}, where it is referred to as Robust PCA.  (A different notion was given that name in~\cite{Brubaker}.)   The formulation also makes sense with different matrix norms, though in this paper, we focus mostly on the Frobenius norm.  Our guarantees also work for entry-wise $\ell_p$ norm error, as we will see.  %Our goal is to design approximation algorithms and study the complexity of this question.  

\subsection{Background and related work}
The statistics literature on robust estimation of parameters is extensive, and we refer to the book of~\cite{HuberNew} for a review.  We now discuss some of the works most relevant to our results.

As mentioned above, the PCA with outliers problem has been well studied in the signal processing community.  The work of~\cite{CandesRobust} shows that by formulating the problem as an optimization problem with appropriate regularization terms, we can efficiently recover the optimal rank $k$ subspace, under certain conditions on the input.  Informally, these conditions say that the inliers must be ``well spread'', which in turn implie the uniqueness of the target subspace (this is similar in spirit to works on matrix completion). The worst-case problem was studied in~\cite{Caramanis}, under a ``coverage'' variant of the objective (in which the goal is to find a rank $k$ subspace that covers as large a fraction of the inlier mass as possible).  The results here do not imply a multiplicative approximation to our formulation, and are thus incomparable.

We also note that while the problem definition naturally suggests $m \ll n$, the problem makes sense when we have a large fraction (approaching $1$) of outliers.  This version was studied in the work of~\cite{HardtMoitra}.  They formulate the problem as follows: given a set of points in $\R^n$ with the condition that an $\alpha$ fraction lie in a $d$-dimensional subspace, the goal is to find the subspace.  Informally, under the condition that the ``outliers'' (the points not in the subspace) are in general position, they develop a simple random sampling algorithm, for the case $\alpha > d/n$.  They also show hardness results for this problem.  Indeed, we note that the hardness results we present in Section~\ref{sec:hardness} are consequences of their approach.\footnote{We thank Amit Deshpande and Ravishankar Krishnaswamy for suggesting this.}  

The main issue with $m \approx n$ is that in the absence of any strong assumptions (along the lines of the above works), the outliers could themselves have an approximate rank-$k$ structure, thus making the problem ill-posed.  For problems such as mean estimation,~\cite{Charikar2017} recently showed that in such cases, something very elegant is possible: we can come up with a small set of {\em candidate} solutions, one of which will be a good solution to the inliers.  This framework was originally introduced in~\cite{Balcan08}, and together with some very interesting new ideas, it has been used to obtain results for classic problems such as learning mixtures of Gaussians under separation~\cite{Gauss1, Gauss2, Gauss3}. 

Finally, as we mentioned above, outlier-robust approximation algorithms are quite well studied for different variants of clustering~(see \cite{Charikar01, Cluster2, Cluster1, RaviShi} and references therein).  For many variants, while the initial algorithms were bi-criteria approximations similar to ours, it turns out that one can actually obtain constant factor approximation algorithms without violating either the bound on the number of centers, or the number of outliers.  This is in contrast to what turns out to be possible for PCA, as we will see. 

\nvspace
\paragraph{Robust algorithmic methods.}  The high level goal in the works above (as well as ours) is the development of algorithmic techniques that work in the presence of outliers.  To this end, linear and semidefinite relaxations have been powerful in identifying the inliers/outliers, and have emerged as a powerful technique.  This is seen in the works of~\cite{MoitraNew, Charikar2017}, the works on clustering mentioned above, and also the extensive literature on semi-random models (see~\cite{FeigeKilian, MMV} and references therein).  Recently,~\cite{SteinhardtCV18} studied the abstract question of when an estimator can be computed robustly, and defined a condition they call {\em resilience}.  This is a property of the inliers that allows estimation (in principle) in the presence of a small fraction of arbitrary outliers.  In this context, our work shows that random sampling based algorithms could be powerful in finding structure in the presence of noise, albeit with weaker (bi-criteria) guarantees.  This is the idea behind heuristics such as RANSAC, which we will now discuss.

\nvspace
\paragraph{Heuristics.}   There have been several heuristics developed specifically for the robust PCA problem, as well as more generally for estimation on noisy data.  One example is the class of ``Lloyd-style'' heuristics, where the idea is to fit a rank-$k$ subspace to the full data set, remove a small number of points that are ``far away'' from the space, and recurse (see~\cite{Cluster1, pursuit}).  Another heuristic the famous RANSAC (random sampling and consensus) paradigm~\cite{ransac}.  The idea here is that if we randomly sample a subset of the points and fit a $k$-subspace, then different samples yield spaces that align ``along the inliers'' but have differing components along the outliers.  Assuming the outliers do not have ``structure'', we can expect that an averaging (consensus) step helps zero in on the inliers. One issue with the above is that the outliers could have an approximate rank $k$ structure.  Intuitively, this is one of the reasons for which we only obtain bi-criteria guarantees. 

\subsection{Our results}\label{sec:results}
In the remainder of the paper, we will write $A = B+N$, where $B$ consists of the inliers (and zeros in the outlier columns), and $N$ contains the outliers (and zeros in the inlier columns).  Thus, in this notation, the problem can be reformulated as that of finding such a decomposition $B+N$ where $N$ has at most $m$ non-zero columns, with the goal of minimizing $\fnorm{B - B_k}$, where $B_k$ is the best rank-$k$ approximation of $B$. Also, we will throughout think of $\eps, \del$ as parameters in the range $(0,1]$.

We present two (incomparable) algorithmic results.  The first is a simple algorithm based on iteratively computing a subspace that captures more and more ``mass'' of the matrix, while throwing away a small number of outliers.  The second algorithm, which is our main contribution, is inspired by {\em adaptive sampling}, an idea that has been successful in obtaining coresets and bi-criteria approximations for problems such as clustering and PCA~\cite{DeshpandeVempala, Silvio}.   To describe the first result, we define the following ``rank-$k$ condition number'':
\[ \Lambda_k := \frac{\fnorm{A}}{\norm{B-B_k}_F^2}. \]

\begin{theorem}\label{thm:adaptive}
There is an efficient algorithm that takes as input a matrix $A$ as above, parameters $k, m, \eps$, and outputs a decomposition $A = B' + N'$ along with a subspace $V$, such that the following properties hold: (a) $N'$ has at most $m \log (\Lambda_k/\eps)$ columns, (b) the space $V$ satisfies the property 
\[ \fnorm{\Pi_{V}^\perp B'} \le (1+\eps) \fnorm{B - B_k},\]
and (c) the dimension of $V$ is at most $k \log (\Lambda_k / \eps)$.  
\end{theorem}

The algorithm above violates the bounds on the number of outliers and the dimension of the space by a factor $\log (\Lambda_k/\eps)$.  Thus it is interesting in the regimes where $m$ is small compared to the number of columns of $A$, and the quantity $\Lambda_k$ is ``not too large''.  For instance, in some practical settings, one might be interested in capturing say 99\% of the mass of a matrix using a small subspace, while excluding a small number of outliers. In such a case, the $\Lambda_k$ is a constant, and all the additional factors are small.

Our second (and main) result has a much better dependence on the parameters.  It has as an additional input a parameter $\delta$, which controls the number of outliers the algorithm outputs.

\begin{theorem}\label{thm:2}
There is an efficient algorithm that takes as input a matrix $A$ as above, parameters $k, m, \eps, \del$, and outputs a decomposition $A = B' + N'$ along with a subspace $V$, such that the following properties hold:  (a) $N'$ has at most $(1+\del) m$ columns, (b) the space $V$ satisfies the property $\fnorm{\Pi_{V}^\perp B'} \le (1+\eps) \fnorm{B - B_k}$, and
(c) $V$ is the span of a subset of the columns of $A$, of size at most
\[ O \left( \frac{k}{\eps^6} \left( \log (n/m) + \frac{2}{\del} \right) \right). \]
\end{theorem}

Thus, the theorem in fact gives a ``column based'' approximation to the space $V$.  Also, if we think of $\eps$ as a constant, the algorithm outputs $O(k \log (n/m) + k/\del)$ columns.  Note that in many cases of interest, we may have $m$ being a small constant factor of $n$.  In such cases, the algorithm roughly outputs only a $O(k/\del)$ dimensional space, while obtaining a $(1+\eps)$ approximation to the objective and violating the number of constraints by a factor $(1+\del)$.

\nvspace
\paragraph{Dependence on $\eps$.}  The drawback in Theorem~\ref{thm:2} is the dependence on $\eps$.  Indeed, the first algorithm, while lossy in many other aspects, has a really good dependence on $\eps$ (which is what makes the algorithms incomparable).  However, we note that even in the noise-free case, obtaining a column-based $(1+\eps)$ approximation to the best rank-$k$ approximation requires $k/\eps$ columns (see~\cite{GuruswamiSinop}).  Thus we cannot hope to get rid of $1/\eps$ entirely.

\nvspace
\paragraph{Technique and extensions.}  It turns out that the key to Theorem~\ref{thm:2} is a modification of a remarkable lemma from~\cite{Silvio}, on finding low-rank approximations under entry-wise $\ell_p$ error by using iterative {\em uniform} sampling.  While the modification (see Lemma~\ref{lem:silvio-lem} and the notes following it) is needed for our main result, it turns out that if we use the original lemma of~\cite{Silvio} in our framework, we obtain the following as a corollary.  Note that the approximation ratio is now $O(k)$ as opposed to $(1+\eps)$ for the case of Frobenius norm.  (A dependence on $k$ turns out to be unavoidable for column-based approximations for $\ell_p$ error even without noise~\cite{Silvio}.)

\begin{theorem}\label{thm:lpnorm}
There is an efficient algorithm that takes as input a matrix $A$ as above, parameters $k, m, \del$, and outputs a decomposition $A = B' + N'$ along with a subspace $V$, such that the following properties hold:  (a) $N'$ has at most $(1+\del) m$ columns, (b) the space $V$ is spanned by $O(k (\log (n/m) + 1/\delta))$ columns of $A$, and (c) the error satisfies 
\[ \err(p, B', V) \le 100 (k+1) \cdot \min_{X \in \R^{d \times k}, Y \in \R^{k \times n}} \norm{B - XY}_p^p , \]
where $\err(p, B', V)$ denotes the minimum $\ell_p^p$ reconstruction error of the columns of $B'$ using  $V$.
\end{theorem}

\nvspace
\paragraph{Limits of approximation. }  It is natural to ask if there needs to be a trade-off between the dimension of the output space and the slack parameter $\delta$.  Furthermore, we can even ask if we can avoid having a slack altogether.  

By a reduction along the lines of the result of~\cite{HardtMoitra}, the following result is quite easy to show.

\begin{theorem}[Informal version of Theorem~\ref{thm:hardt-moitra}]
Under the small set expansion conjecture with suitable parameters, for any constant $C> 0$, there is no polynomial time algorithm that can obtain a multiplicative factor approximation to the objective, while returning a $Ck$ dimensional subspace, and excluding at most $(1+\delta)m$ outliers,  for small enough constant $\delta$. 
\end{theorem}

This rules out the possibility of finding a $Ck$-dimensional subspace for arbitrarily small $\delta$.  A very similar reduction, but from the smallest $r$-edge subgraph problem (see Section~\ref{sec:hardness:2}) implies that if we wish to have $\delta = 0$, then the dimension of the subspace output must be at least $k \cdot n^{\Omega(1)}$.  See Corollary~\ref{thm:smes} for the formal statement. We remark once more that these hardness results indicate that ``pure'' approximations (as can be obtained for clustering) are impossible for PCA.

\paragraph{Open problems, directions.}  While bi-criteria guarantees are unavoidable in the worst case, it is interesting to see if simple iterative algorithms like the ones we proposed can be shown to {\em recover} the $k$-PCA subspace of the inliers under appropriate assumptions.  A starting point would be assumptions similar to the ones of Donoho et al.  Next, the dependence on $\delta, \epsilon$ in our sampling based algorithm are possibly sub-optimal.  It would be interesting to see if more sophisticated algorithms can give better guarantees.  More broadly, it would be interesting to show more guarantees for heuristics such as RANSAC and algorithms inspired by them for other problems involving outliers.  

\section{Notation and preliminaries}
We start with some basic matrix notation we use.   Let $A$ be a $d \times n$ matrix.  Throughout the paper, we write $A_k$ to refer to the best rank-$k$ approximation of $A$ (thus it is also a $d \times n$ matrix).  For a subset $T$ of the column indices ($T \subseteq [n]$), we denote by $A\subb{T}$ the $d \times |T|$ sub-matrix of $A$ formed by the columns indexed by $T$.

Next, given an integer $k$, we denote by $\err_k (A)$ the error in the best rank-$k$ approximation of $A$.  Specifically, $\err_k(A) = \fnorm{A- A_k}$.  Also, for a set of vectors $W$, their linear span is denoted $\spn(W)$.  Finally, the projection matrix orthogonal to the space orthogonal to $\spn(W)$ will be denoted by  $\Pi_W^\perp$.   So also, for a {\em subspace} $W$, we abuse notation slightly and denote by $\Pi_W^\perp$ the projection matrix to the space orthogonal to $W$.

\nvspace
\paragraph{Guessing the optimum.}  In all our algorithms, we assume that we have a guess for the optimum error (error in the low-rank approximation of $B$), up to a multiplicative factor of $(1+\eps)$.   A fairly straightforward argument shows that we can always come up with a polynomial number (in the input complexity) of guesses, one of which is accurate.  This is shown in Appendix~\ref{sec:opt-guess}.

\section{Iterative SVD}
We now present our first algorithm.  It involves repeatedly computing the best low rank approximation, while potentially throwing away some points as outliers.  This will establish Theorem~\ref{thm:adaptive}.  %In Section~\ref{sec:adaptive-norm}, we observe that the SVD step can be replaced by norm sampling as in~\cite{FriezeKannan}.
Let us start with an informal description of the algorithm.

\nvspace
\paragraph{Algorithm outline.} At each step $j$, we have a subset $S_j$ of the initial column vectors, and a subspace $V_j$. Let $N = |S_j|$ and let $u_1, u_2, \dots, u_N$ denote the projections of the columns in $S_j$, orthogonal to the space $V_j$. The aim is to either (a) find a new subspace $V_{j+1}$ that captures a significantly larger fraction of the total mass than $S_j$, or (b) remove a set of at most $m$ columns from $S_j$ and mark them as outliers. In either case, we show that the total ``uncaptured'' mass remaining drops by a constant factor.  This allows us to bound the number of iterations, thus giving the guarantees of Theorem~\ref{thm:adaptive}.

The algorithm is formally stated below (Algorithm~\ref{alg:1}).  As discussed in Section~\ref{sec:opt-guess}, it assumes that we have a guess $\xi$ for the optimum error. 

\begin{algorithm}[!h]
\caption{Iterative SVD}
\label{alg:1}

 {\bf Input:} Matrix $A \in \R^{d \times n}$, guess $\xi$ for the optimum error, parameter $m$ (bound on \# outliers), and accuracy parameter $\eps$.\\
{\bf Output:} A subspace $V$, and a set $S$ of inliers. %$O(k \log \Lambda_k/(\epsilon \delta))$ columns of A

\begin{algorithmic}[1]
%
%Let $S_{0}$ be the set of all initial column vectors, $V_{0}$ be an empty set representing a subspace
%
%	Let $u_1, u_2, \dots, u_N$ denote the projections of the columns in $S_j$ orthogonal to the space $V_j$.
%
%    
\State Initialize $ V_0 = \emptyset$, $S_0 = \text{cols}(A)$, and $j = 0$.
\While{total squared projection of $S_j$ onto the space orthogonal to $V_j$ is $\ge (1+\eps) \xi$} 
	\State Let $u_1, u_2, \dots, u_N$ be the projection of the columns in $S_j$ orthogonal to $V_j$.
	\State Define $\mu_j := \sum_i \norm{u_i}^2$.
	\State Let $T_j$ be the $m$ largest (by length) vectors among $\{u_i\}_{i=1}^N$.
	\If{($\sum_{u \in T_j} \norm{u}^2 \geq \frac{1}{2} (\mu_j -\xi)$ )} 
	\State $S_{j+1} := S_j \setminus T_j$, and $V_{j+1} = V_j$.
	\Else
	\State Let $V_{j+1}$ be the rank-$((j+1)k)$ SVD for $S_j$.
	\State Set $S_{j+1} = S_j$.
	\EndIf
	\State $j \leftarrow j+1$
\EndWhile
\State \Return $V_j, S_j$
\end{algorithmic}
\end{algorithm}

\begin{lemma}\label{lem:adapt-sample}
In every iteration of the algorithm, the total mass of the inliers reduces significantly.  More precisely, we have $\mu_{j+1} \le (\mu_j + \xi)/2$.   
\end{lemma}
\begin{proof}
We consider both the cases in the algorithm.

\nvspace
\paragraph{Case 1.} $\sum_{u \in T_j} \norm{u}^2 \geq \frac{1}{2} (\mu_j -\xi)$. 

In this case, we remove $T_j$ from the set of inliers, and thus
\[ \mu_{j+1} = \mu_j - \sum_{u \in T_j} \norm{u}^2 \le \frac{\mu_j + \xi}{2}. \] 

\nvspace
\paragraph{Case 2.}  $\sum_{u \in T_j} \norm{u}^2 < \frac{1}{2} (\mu_j -\xi)$. 

Let us denote the set of inlier columns by $\cI$.  Let us argue about the error in the best rank-$(j+1)k$ approximation of $S_j$.  To do so, we consider the space $V' = V_{j} + V^*$, where $V^*$ is the rank-$k$ SVD space of $\cI$ (this is the optimal subspace that we are after).   Now, let us consider the projection of the columns of $S_j$ orthogonal to $V'$.  For the columns $S_j \cap \cI$, the total error has to be $\le \xi$, because by assumption, projecting $\cI$ orthogonal to $V^*$ has this error.  Next, the projection of $S_j \setminus \cI$ orthogonal to $V'$ must be smaller than (or equal to) the projections orthogonal to $V_j$.  As $|S_j \setminus \cI|\le m$ (there are at most $m$ outliers), their projections orthogonal to $V_j$ can be bounded by $\sum_{u \in T_j} \norm{u}^2$ (as $T_j$ contained the $m$ largest vectors by length).  This allows us to bound the total error by
\[  \xi + \frac{\mu_j - \xi}{2} = \frac{\mu_j + \xi}{2}. \]
The best $(j+1)k$ dimensional subspace will result in an error only better than the above, and thus the lemma follows.\footnote{We note that a more ``natural'' algorithm is to add the top-$k$ SVD space of the $u_i$ vectors to the current space in step 9.  This  is closer to adaptive sampling paradigm (see~\cite{DeshpandeVempala}), but it runs into the issue that it is not clear the projection of the $u_i$ corresponding to inliers reduces to $\xi$.  This stems from the fact that for two spaces $V, W$, $\Pi_V u$ could be smaller in length than $\Pi_V (\Pi_W u)$. }
\end{proof}

The lemma above immediately implies that (assuming that the guess of $\xi$ is an upper bound on the optimum), the algorithm runs for at most $\log (\fnorm{A}/\eps \xi)$ iterations.  This is because the gap between $\mu_j$ and $\xi$ drops by a factor at least 2 in each iteration.  Thus, by searching over all the possible $\xi$ and by the comment below, Theorem~\ref{thm:adaptive} follows.

\nvspace
\paragraph{Validating the guess	of $\xi$}  We note that the above algorithm works whenever $\xi$ is an upper bound on the optimum.  If it is lower than the optimum, then in one of the iterations, we may not see a drop in $\mu_{j+1}$ as guaranteed by Lemma~\ref{lem:adapt-sample}.  Thus, we can test for this as the algorithm proceeds and output FAIL if we do not see a drop. 

\section{Iterative uniform sampling}
We next present our main algorithm for the PCA with outliers problem.  We will start with a sampling lemma that is at the heart of the algorithm.  This lemma applies to the case when there are no outlier columns.  As mentioned earlier, the lemma is a variant of a lemma from~\cite{Silvio}, which applies to $\ell_p$ norms and has additional factors of $k$.  In Section~\ref{sec:sample-outliers}, we show how the lemma can be used even in the presence of outliers, thus establishing Theorem~\ref{thm:2}.

\subsection{Sampling without outliers}
The first lemma is simply about low rank approximation via columns (without any outliers). 

\begin{lemma}\label{lem:silvio-lem}
Let $A \in \R^{d \times n}$, and let $\eps \in (0, 1]$ be any parameter. Let $S$ be a uniformly random subset of $[n]$ of size $s$, where $s\ge 4k/\eps^2$.  Then, with probability at least $\eps^2/8$, there exist a set of $\eps^2 n/8$ columns of $A$, whose projection orthogonal to the column span of $A\subb{S}$ is upper bounded by $(1+\eps) \norm{A- A_k}_F^2 / n$.
\end{lemma}

\nvspace
\paragraph{Note.}  The lemma is quite surprising: for say $\eps = 1$, it says that a {\em uniformly random} subset of $4k$ columns covers a constant fraction of the columns up to a constant times the $k$-SVD error.  Such a guarantee is not even clear for norm-based sampling (see~\cite{FriezeKannan}).  So, while uniform sampling does not necessarily give a low {\em total error} in expectation, it ends up giving small error for a constant fraction of columns.  

We use the same rough outline as the proof of~\cite{Silvio}.  There are, however, two main differences.  First, as we need a column-based $(1+\eps)$ guarantee on rank-$k$ approximation, we appeal to the results of~\cite{GuruswamiSinop}, instead of a weaker $O(k)$ bound used in~\cite{Silvio} for the $\ell_p$ norm.  Second (and more significant), their proof uses a simple union bound over failure probabilities.  This does not suffice for a $(1+\eps)$ approximation, as two of the events are low probability (roughly $\eps$).  We observe that these events are effectively independent, and so we obtain a constant probability of success. .

\begin{proof}
Let us denote $s = 4k/\eps^2$.  Let $T$ be a random subset of $[n]$ of size $s+1$.  We think of sampling $S$ as first sampling $T$ and then removing a random element $u \in T$. For convenience, let us write
\[ \theta = \frac{\fnorm{A - A_k}}{n}.\]

First, let us call a subset $T$ {\em good} if $A\subb{T}$ has a small error rank-$k$ approximation. Concretely, $T$ is said to be good if $\err_k (A\subb{T}) \le (1+\eps) \cdot |T| \theta$. Now, if $\Pi_k$ is the projection matrix onto the space orthogonal to the best rank $k$ approximation for the full matrix $A$, we have $\E [\sum_{u \in T} \norm{\Pi_k u}^2 ] = |T| \theta$ (where the expectation is over the choice of $T$).  Thus, by Markov's inequality, we have
\begin{equation}\label{eq:pr-good}
\Pr \left[ \sum_{u \in T} \norm{\Pi_k u}^2 > (1+\eps) |T| \theta \right] \le \frac{1}{1+\eps} \le 1 - \frac{\eps}{2}.
\end{equation}

Thus, as the best rank-$k$ approximation for $A\subb{T}$ can only have a smaller error, we have that $T$ is good with probability at least $\eps/2$. Next, we show the following claim.

\begin{claim}\label{claim:1}
For any $T$ of size $\ge 4k/\eps^2$, if we form $S$ by randomly removing a $u \in T$, then with probability at least $\eps/2$ (over the choice of $u$), we have 
\begin{equation}\label{eq:claim1}
\norm{\Pi_S^\perp u}^2 \le (1+\eps) \cdot \frac{err_k(A\subb{T})}{|T|}.
\end{equation}
\end{claim}

To show this claim, we first appeal to the existence of good column-based low-rank approximations. Specifically, Guruswami and Sinop showed the following.
\begin{theorem}[Guruswami, Sinop~\cite{GuruswamiSinop}]\label{thm:guruswami}
For any matrix $B$ and parameter $k$, there exist a subset $W$ of at most $k/\eps$ columns of $B$ with the property that
\[ \norm{\Pi_W^\perp B}_F^2 \le (1+\eps) \cdot \err_k (B).\]
\end{theorem}

We apply the result to the matrix $A\subb{T}$, and let $W$ denote the set of columns guaranteed by the theorem. Now, the probability that a random column $u$ of $A \subb{T}$ belongs to $W$ is at most $|W|/|T|$, which is at most $\eps/2$, by our choice of $s$. Thus, the probability that $S \supset W$ is at least $1 - \eps/4$.

Next, we also have that for a random column $u$ of $A\subb{T}$, the expected value 
\[ \E[ \norm{\Pi_W^\perp u}^2 ] = \frac{\fnorm{\Pi_W^\perp A\subb{T}}}{|T|} \le (1+\eps) ~\frac{\err_k (A\subb{T})}{|T|}. \]
Once again, by Markov's inequality, the probability that $\norm{\Pi_W^\perp u}^2$ is bounded by $(1+\eps)$ times the RHS is at least $\eps/2$. Thus by a union bound, with probability at least $\eps/4$, we have this condition, as well as the event $S \supset W$. In this case, we clearly have $\norm{\Pi_S^\perp u} \le \norm{\Pi_W^\perp u}$, and this completes the proof of Claim~\ref{claim:1}.

Now, consider the bipartite graph in which the left side consists of all $(s+1)$-tuples of $[n]$ and the right side consists of all $s$-tuples of $[n]$.  We place an edge between $T$ and $S$ if (a) $T$ is good, and (b) $u = T \setminus S$ satisfies Eq.~\eqref{eq:claim1}. The claim above, together with Eq.~\eqref{eq:pr-good} (which lower bounds the probability of $T$ being good), imply that the total number of edges is at least
\[ \frac{\eps}{2} \binom{n}{s+1} \frac{\eps}{4} (s+1) = \frac{\eps^2}{8} (n-s) \binom{n}{s}. \]
Note that $n-s$ is the maximum number of edges that could be incident to a vertex on the right.  Thus, we conclude that at least an $\eps^2/8$ fraction of the vertices on the right have degree at least $\eps^2 (n-s)/8$.  Since an edge implies that (a) $T$ is good and (b) Eq.~\eqref{eq:claim1} holds, the conclusion of the lemma follows.
\end{proof}

\subsection{Incorporating outliers}\label{sec:sample-outliers}
Next, suppose the set of columns contains (at most) $m$ outliers. We then consider Algorithm~\ref{alg:2}.  The analysis will use the value of the average error per column in the optimum solution, namely $\theta :=  \fnorm{B - B_k}/(n-m)$.

\begin{algorithm}[!h]
\caption{Iterative sampling with outliers}
\label{alg:2}
{\bf Input:  } Matrix $A \in \R^{d \times n}$, parameters $m, k, \del, \eps$, guess $\theta$ for optimum error per column.\\
{\bf Output: } A set of outliers $\cO$ and a set of columns $V$ of $A$.

\begin{algorithmic}[1] 
\Procedure{Select}{$A, m, k, \eps, \del$}
	\State Initialize $\cO = V = \emptyset$.  Define $N = \text{\#cols}(A)$. 
	\State Define $n_0 = \frac{\alpha}{\alpha - 1} \cdot \frac{8k}{\eps^3}$, where $\alpha = N/m$.
	\If{$N < n_0$}
	\Return add all the columns of $A$ to $V$ and return $(\cO, V)$
	\ElsIf{ $N \le (1+\del) m$}
	\State add all the columns of $A$ to $\cO$ and return $(\cO, V)$
	\Else
	\For{$16 \log n/\eps^2$ iterations}
		\State Let $R \leftarrow$ a uniformly random sample of $n_0$ columns of $A$
		\State Let $\hat{A}$ be the set of $\eps^3 (N-m)$ columns of $A$ that have the least projection to $\Pi_{R}^\perp$ 
		\State Let $X = \fnorm{\Pi_{R}^\perp \hat{A}}$
		\State  If $X$ is smaller than the least such quantity so far, set $A^* = \hat{A}$, $R^* = R$
	\EndFor
	\State Mark the columns $A^*$ as {\em covered}
	\State Let $(\cO', V')$ be the output of the recursive call $\textsc{Select} (A \setminus A^*, m, k, \eps, \del, \theta)$
	\State \Return $(\cO', R^* \cup V')$
 	\EndIf
\EndProcedure
\end{algorithmic}
\end{algorithm}

Define $n' = n - m$.  Let $z_1, z_2, \dots, z_{n'}$ be the rank $k$ approximation errors of the columns $B_i$ of $B$. (So $\sum_i z_i = n' \theta$.) We will assume (without loss of generality) that $z_i$ are in increasing order.  

\nvspace
\paragraph{Proof outline.}  Consider the first iteration of the algorithm.  We claim that $\fnorm{\Pi_{R^*} A^*} \le (1+\eps) (z_1 + z_2 + \dots + z_{\eps n'})/(\eps n')$ with probability at least $1 - \frac{1}{n^4}$.  This is because with high probability, one of the candidates for $R$ would have chosen $4k/\eps^2$ columns {\em among} $B_1, \dots, B_{\eps n'}$, and then we can apply Lemma 6. Thus, the error for the first $n' \eps^3$ columns covered is at most the average error for the first $\eps n'$ columns.  Subsequently, we will be able to bound the error in each step in terms of the smallest $\eps$ fraction of the $z_i$'s that remain. 

\begin{lemma}\label{lem:inductive}
Consider the procedure $\textsc{Select}(A', m, k, \eps, \delta)$, and suppose that $N := \text{\#cols}(A')$ satisfies $N \ge  \max\{ n_0, (1+\del) m \}$, where $n_0$ is defined in step 3 of the algorithm.  Let $B' = A' \cap B$, i.e., the set of inliers in $A'$, and let $y_1, y_2, \dots, y_{|B'|}$ denote the values $z_i$ for $i \in B'$.   Assume w.l.o.g. that $y_i$ are in increasing order.  Then with probability $\ge 1-\frac{1}{n^4}$, the sets $R^*, A^*$ chosen by the procedure satisfy
\begin{equation}
\text{ for all $u \in A^*$, } \fnorm{ \Pi_{R^*}^{\perp} u} \le \frac{\sum_{i=1}^{\eps |B'|} y_i}{\eps |B'|}. \label{eq:inductive}
\end{equation}
\end{lemma}
\begin{proof}
For each iteration of the loop (8-13) of the algorithm, we show that the probability of the chosen $R, \hat{A}$ satisfying the bound in Eq.~\eqref{eq:inductive} is at least $\eps^2/2$.  The conclusion of the lemma then follows immediately.

First, note that for $\alpha = N/m$, we have $|B'|/|A'| \ge (\alpha - 1)/\alpha$.  Let us also denote by $Q$ the set of $\eps |B'|$ columns of $B'$ that have the smallest $z_i$ values.   Now, the expected size of $R \cap Q$ is 
\[ \frac{|R|}{|A'|} \cdot \eps |B'| \ge \frac{8k}{\eps^2}. \]
Thus the probability that this is $\ge 4k/\eps^2$ is at least $1/2$. (Formally, this follows from Hoeffding's inequality, which also applies to sums of random variables without replacement.)

Conditioned on $|R \cap Q| \ge 4k/\eps^2$, we can apply Lemma~\ref{lem:silvio-lem} to conclude that with probability $\ge \eps^2$, at least an $\eps^2$ fraction of the columns in $Q$, the projection orthogonal to $\spn{}(R)$ is bounded by $\sum_{i \in Q} z_i / |Q|$.  Note that by definition, this is precisely the RHS of Eq~\eqref{eq:inductive}.   Thus the probability that the chosen $R, \hat{A}$ satisfy~\eqref{eq:inductive} is at least $\eps^2/2$, and this completes the proof of the lemma.
\end{proof}

\begin{lemma}\label{lem:all-else}
When the algorithm terminates, the total error incurred is bounded by $\frac{1+\eps}{1-\eps} n' \theta$, with high probability.   Further, the depth of the recursion is upper bounded by $\frac{\log (n/(\delta m))}{\eps^3}$.  The total number of columns chosen is at most \[ \frac{16 k}{\eps^6} \left(  \log (n/m) + \frac{2}{\del} \right).\]
\end{lemma}
\begin{proof}
Using Lemma~\ref{lem:inductive}, we can analyze the overall error as follows.  Since the success probability in each iteration is $1-\frac{1}{n^4}$ we assume henceforth that the conclusion of Lemma~\ref{lem:inductive} holds for all the recursive calls.  Let us divide the indices $[n']$ (from left to right) into groups of size $\eps n', \eps (1-\eps)n', \eps (1-\eps)^2 n', \dots$. Let us call the groups $G_1, G_2, \dots$, and let $E_i$ denote $\sum_{j \in G_i} z_j$. Now, by the lemma, until the algorithm marks $\eps n'$ columns as covered, we have that the average error in the marked columns is at most $(1+\eps) E_2 / |G_2|$, w.h.p. (Note that the set of columns marked by the algorithm can be quite different from $G_1$; but this will only help the argument, which only requires that an $\eps$ fraction of the uncovered columns has average error $\le E_2/|G_2|$.)  Likewise, until the next $\eps (1-\eps) n'$ columns are marked covered, the average error is $\le (1+\eps) E_3/|G_3|$. This continues until the number of unmarked columns falls below $k/\eps^3$, at which point the error will be zero, as all the columns will be picked.

Thus, we cover the first $|G_1|$ columns with average error $\le (1+\eps) E_2/|G_2|$, the next $|G_2|$ columns with average error $\le (1+\eps) E_3/|G_3|$, and so on. And when $|G_i|$ gets to $k/\eps^3$ (or $\delta m$), the error is zero and the procedure stops. Since $|G_i|/|G_{i+1}| = 1/(1-\eps)$, we can bound the total error by
\[  (1+\eps) \frac{E_2 + E_3 + \dots}{1-\eps} \le \frac{(1+\eps) n'\theta}{1-\eps}.\]
This completes the proof of the error bound.

For the depth of the recursion, note that we always mark precisely $\eps^3(N- m)$ columns as marked.  The procedure terminates when $N-m$ is $\max \{n_0, (1+\del)m\}$, and this gives the desired bound.

To bound the number of columns chosen, we need to analyze the sum of $n_0$ over the recursive calls.  We note that as long as $\alpha \ge 2$, we have $n_0 \le 16 k/\eps^3$.  By the argument above, the number of recursive steps needed to reach $\alpha = 2$ is at most $(1/\eps^3) \cdot \log (n/m)$.  This bounds the number of columns chosen until this step by $(16 k/\eps^6) \log (n/m)$.  Next,  as $(\alpha - 1)$ drops from $2^{-j}$ to $2^{-j-1}$,  we end up with  $n_0 \le 2^{j+1} \cdot 8k/\eps^3$  columns in each call to {\sc Select}.   The number of recursive calls necessary for this drop in $\alpha$ is at most $1/\eps^3$.  Thus, summing over $j$ from $0$ through $\log (1/\delta)$, we have that the number of columns chosen is at most
\[  \sum_{j = 0}^{\log(1/\del)}  2^{j+1} \frac{8k}{\eps^6} \le \frac{32k}{\eps^6 \del}.\]
This completes the proof of the lemma.
\end{proof}

Theorem~\ref{thm:2} now follows easily. 
\begin{proof}[Proof of Theorem~\ref{thm:2}]
The desired bound on the number of columns, as well as the bound on the approximation ratio and the number of outliers all follow from Lemma~\ref{lem:all-else}.  
\end{proof}

\subsection{Entry-wise $\ell_p$ error}
We notice that the algorithm above can easily be adapted to the case in which we care about the entry wise $\ell_p$ error. Again, we have a sampling lemma in the setting without outliers.  As mentioned earlier, the following lemma was shown in~\cite{Silvio}.  
\begin{lemma}[Lemma 6 of~\cite{Silvio}]
Let $A \in \R^{d \times n}$, and let $A_{k,p}$ be a minimizer of $\norm{A - X}_p$ over rank-$k$ matrices $X$.  Let $S$ be a uniformly random subset of $[n]$ of size $2k$.  Then, w.p. at least $1/10$, there exists a set of $n/10$ columns of $A$ whose optimum $\ell_p^p$ reconstruction error using the columns of $A\subb{S}$ is at most $100 (k+1) \norm{A - A_{k,p}}_p^p /n$.
\end{lemma}

We can now apply the reasoning from earlier, along with a modified algorithm, in which we treat a column as ``covered'' if the $\ell_p^p$ reconstruction error using the columns chosen is at most $100 (k+1)/n$ times the guess for the optimum.  One component we need here is to be able to compute the optimum $\ell_p^p$ reconstruction error.  This can be done via a convex program for any $p \ge 1$.  We refer to~\cite{Silvio} for the details.  By following the proof earlier, we obtain Theorem~\ref{thm:lpnorm}.  We omit the details.

The number of columns is now $O \left( k( \log (n/m) + 1/\delta) \right)$.

\section{Limits of approximation}\label{sec:hardness}
We make some simple observations about the computational complexity of the problem of PCA with outliers.  
First, we consider the consequences of a reduction due to Hardt and Moitra~\cite{HardtMoitra} to our setting. This yields a strong impossibility result assuming the small set expansion (SSE) conjecture (see~\cite{sse1}). Then, we show that if we do not allow a slack in the number of outliers, then we cannot even hope to find a ``reasonably small'' dimensional subspace with an error within any multiplicative factor.

\paragraph{Notes on our reductions.} Before proceeding, we note the following
\begin{enumerate}
\item In the reductions, the number of outliers is very close to the total number of columns. While this is intuitively not the regime of ``practical interest'', it is easy to see that by padding $O(n)$ copies of a vector orthogonal to all the ones produced in the reduction, we (a) obtain the setting in which the number of outliers is a small constant fraction of the total number of columns, and (b) have the same lower bounds (because a change of $\pm 1$ to the target dimension does not matter in our proofs). Intuitively, these are cases in which one of the components is easy to find, and it becomes trickier to find the others. These are precisely the type of pathologies avoided by the ``well spread'' assumptions in~\cite{CandesRobust, Chen}.
\item Next, we note that both the reductions are in the case when the optimum error is zero.  Thus, our lower bounds  imply that we cannot obtain {\em any} multiplicative approximations under the corresponding assumptions. 
\end{enumerate}

\subsection{Reduction from SSE}
Hardt and Moitra~\cite{HardtMoitra} give a reduction from small set expansion (SSE) to {\em robust subspace recovery}, a problem closely related to PCA with outliers. Let us start by recalling the $\sse (\delta, \eps)$ problem.  We are given a $\Delta$-regular graph $G = (V, E)$ on $n$ vertices, and the goal is to distinguish between the following two cases
\begin{quote}
YES:  there exists a set $S$ of size $\delta n$ with $\Phi(S) \le 1/2$.\footnote{As is standard, $\Phi(S)$ denotes the conductance of $S$, namely the quantity $\frac{E(S, \overline{S})}{\Delta \cdot \min\{|S|, n-|S|\}}$. }\\
NO:  every set of size $\le \delta n$ has expansion $\ge 1 - \eps$. 
\end{quote}
We note that typically the $\sse$ problem is stated with the YES case having $\Phi(S) \le \eps$.  The version above is clearly at least as hard.  The ``SSE conjecture'' \cite{sse1}  states that for any $\eps > 0$, there is a small enough $\delta > 0$ such that $\sse(\delta , \eps)$ is hard (for polynomial time algorithms).

Now, the reduction of~\cite{HardtMoitra} constructs a collection of vectors in $\R^{|V|}$, one for each edge.  The edge $\{i,j\}$ corresponds to the vector $\e_i + \e_j$, where $\e_i$ denotes the unit vector of the standard basis.  For an edge $f$, we denote this by $v(f)$.  The main lemma behind their reduction is the following:
\begin{lemma}[Hardt, Moitra~\cite{HardtMoitra}]\label{lem:edge-subspace}
Let $E'$ be a subset of edges, and let $n'$ be the total number of vertices that these edges are incident to. Then we have
\[  n'/2 \le  \dim (\spn(\{ v(f) : f \in E'\})) \le n'. \]
\end{lemma}
The proof is simple, and proceeds as follows.  First, the upper bound is trivial, as all the vectors are spanned by the vectors in the standard basis corresponding to the vertices incident to $E'$.  The lower bound follows from arguing that for any spanning tree $T'$, the vectors $v(f)$ corresponding to the edges of $T'$ are linearly independent.  This can be shown by way of contradiction.  Suppose there exist $\alpha_f$ such that $\sum_{f \in T'} \alpha_f v(f) = 0$.  The edges corresponding to the non-zero coefficients form a forest.  Now consider any {\em leaf} $j$, i.e., a vertex that is incident to precisely one edge in the forest (which has to exist).  Then $\iprod{\e_j, \sum_{f} \alpha_f v(f)} \ne 0$, which is a contradiction.

The lemma implies the following about PCA with outliers.
\begin{theorem}\label{thm:hardt-moitra}
Suppose $\delta, \eps$ are constants such that $\sse(\delta, \eps)$ is hard.  Then, even in the zero error case of PCA with outliers, there is no efficient algorithm that can find a subspace of dimension $k / 6\sqrt{\eps}$ containing all but $(1+\delta/4)m$ of the points.  (Recall $m$ is the prescribed number of outliers.)
\end{theorem}

The theorem says that even if we allow a $(1+\delta)$ fraction more outliers, we cannot have any constant factor approximation, assuming SSE. This, in a sense, is why the dimension of the space returned in Theorem~\ref{thm:2} has to have a dependence on $\del$. 

\begin{proof}
We consider the reduction as above, and set the upper bound on the number of outliers to
\[ m := \frac{ n\Delta}{2} - \frac{\delta n \Delta}{2} = \frac{n\Delta}{2} \left( 1 - \frac{\delta}{2} \right). \] 
In the YES case of SSE, as discussed above, the space spanned by the standard basis vectors corresponding to the non-expanding set contains all the edges incident to that set, and hence we obtain the desired bound on the number of outliers.  

In the NO case, by the choice of parameters, violating the number of outliers by a factor $(1 + \delta/4)$ still means that we must have at least $\Delta \cdot n \delta /8$ {\em inlier} columns.  Suppose there is a space of dimension $\le c \delta n$ that contains this many columns.  Then, by Lemma~\ref{lem:edge-subspace}, there must be a set of $2 c \delta n$ that has at least $\Delta \cdot n \delta/8$ edges.  Now a uniformly random subset of $\delta n$ of these vertices has (in expectation) at least $\Delta \cdot n \delta / 32 c^2$ edges.  Now, if every set of size $\delta n$ has expansion at least $1-\eps$ (since we are in the NO case), then every such set contains at most $\eps \Delta \cdot n \delta $ edges.  Thus, we have that 
\[ \frac{1}{32c^2} \le \eps, \quad \text{which implies } c\ge \frac{1}{6  \sqrt{\eps}}.\]

This completes the proof of the theorem.
\end{proof}

\subsection{Reduction from smallest $r$-edge subgraph}\label{sec:hardness:2}
If we do not allow {\em any} slack in the number of outliers, we show that the situation is rather hopeless. The smallest $r$-edge subgraph (S$r$ES)~\cite{Chlamtac} problem is the following: given a graph $G = (V, E)$ on $n$ vertices, and a parameter $r \le |E|$, the goal is to find an induced subgraph $H$ with the smallest number of {\em vertices}, that has at least $r$ edges. The problem is closely related to (and in some sense is a {\em dual} of) the densest $k$-subgraph problem (DkS)~\cite{Dkspaper}.  \cite{Chlamtac} obtained the best known approximation algorithms for the S$r$ES problem, with an approximation factor $O(n^{2 - \sqrt{3} + \eps})$, for any $\eps > 0$.  The complexity of the problem is also closely related to that of DkS, and indeed, the following is believed to be true (see~\cite{Chlamtac, dksgaps} and references therein):
\begin{conjecture}\label{conj:smes}
The smallest $r$-edge subgraph problem is NP-hard to approximate to a factor $n^{c}$, for some absolute constant $c > 0$.  Specifically, there exist parameters $r, d$ such that it is NP-hard to distinguish between
\begin{quote}
YES:  there exists an induced subgraph on $d$ vertices with $r$ edges. \\
NO: the smallest induced subgraph containing $r$ edges has at least $d n^c$ vertices.
\end{quote} 
\end{conjecture}

Now, the following is a simple corollary.
\begin{corr}\label{thm:smes}
Consider the PCA with outliers problem in which the algorithm is constrained to ignore at most $m$ outliers (without any slack).  Then, assuming Conjecture~\ref{conj:smes}, there is a constant $c >0$ such that it is NP-hard to find a $k n^c$ dimensional subspace that results in a multiplicative approximation to the objective.
\end{corr}
\begin{proof}
We can use the same argument as before, and give a reduction from a gap version of S$r$ES.  If we set the number of outliers to be $|E|-r$, then in the YES case, there is a subspace of dimension $d$ containing $r$ edges, while in the NO case, any such subspace must have a dimension at least $d n^c/2$ (using Lemma~\ref{lem:edge-subspace}).
This completes the proof.
\end{proof}

\bibliographystyle{alpha}
\bibliography{bhaskara}

\newpage
\appendix
\section{Guessing the optimum}\label{sec:opt-guess}
In all our algorithms, we assume that we have a guess for the optimum error, up to a multiplicative factor of $(1+\eps)$.  Note that the error is $\fnorm{B-B_k}$, as defined in the introduction.  We now argue that we can always obtain a polynomial number of guesses, one of which is accurate.  This follows immediately if we can show that the error lies in a range $[L, U]$, where $U/L$ is at most $\exp (\poly (n, b))$, where $b$ is the total bit complexity of the input. (This is because we can discretize the range into $\poly (n, b)/\eps$ intervals using multiples of $(1+\eps)$.)

This is not immediate in our setting because $\fnorm{B- B_k}$ can actually be zero.  But we can show that if the error is {\em non-zero}, it can be lower bounded by $\exp(-\poly (n, b))$.  Such results are quite well-known (see~\cite{Golub}), but we sketch a short proof below. 

\paragraph{Exponential range for the optimum.}   The key claim is that if a $d \times k$ matrix $C$ has linearly independent columns, it has $\sigma_{\min} \ge \exp(-\poly (n, b))$, where $b$ is the total bit complexity of $C$.  This can be proved as follows. Note that $\sigma_{\min}$ is the smallest eigenvalue of $C^T C$.  Now the characteristic polynomial of $C^T C$ has coefficients that are all at most $\exp(\poly(n, b))$, as they are appropriate determinants.  Suppose the polynomial is $\sum_{i=0}^k a_i \lambda^i$.  By linear independence, we know that $a_0 \ne 0$. Further, $a_0$ is the determinant of $C^T C$, and hence is bounded from below by $\exp(-\poly(n, b))$; this is because the sum of a set of numbers of a certain precision is either zero or is at least the ``least count'' of that precision.  

Now, the magnitude of the smallest non-zero real root of any polynomial can be bounded from below by $|a_0|/(|a_0|+ |a_1| + \dots + |a_d|)$.  As the $a_i$ are all at most $\exp(\poly(n, b))$, the desired conclusion follows.

\end{document}